\documentclass[letterpaper, 10 pt, conference]{ieeeconf}  

\IEEEoverridecommandlockouts                              
\usepackage{psfrag}
\usepackage{graphicx}        
\usepackage{algorithmic}                            
\usepackage{amsmath,amsxtra,amssymb}
\usepackage{algorithm2e}
                             
\newtheorem{theorem}{Theorem}

\newtheorem {proposition} {Proposition}

\newtheorem{remark}{Remark}

\newtheorem{assumption}{Assumption}

\title{\LARGE \bf
Fictitious Play with Time-Invariant Frequency Update \\ for Network Security
}

\author{Kien C. Nguyen, Tansu Alpcan, and Tamer Ba\c{s}ar
\thanks{This research was supported by grants from the Deutsche Telekom Laboratories (Berlin, Germany) and the Boeing Company.}
\thanks{Tamer Ba\c{s}ar and Kien C. Nguyen are with the Department of Electrical and Computer Engineering and the Coordinated Science Laboratory, University of Illinois at Urbana-Champaign, 1308 W Main St., Urbana, IL 61801, USA
        {\tt\small basar1@illinois.edu, knguyen4@illinois.edu}}%
\thanks{Tansu Alpcan is with the Deutsche Telekom Laboratories and the Technical University of Berlin, Ernst-Reuter-Platz 7, D-10587 Berlin, Germany
        {\tt\small tansu.alpcan@telekom.de}}%
}

\begin{document}

\maketitle
\thispagestyle{empty}
\pagestyle{empty}

\begin{abstract}

We study two-player security games which can be viewed as sequences of nonzero-sum matrix games played by an Attacker and a Defender. The evolution of the game is based on a stochastic fictitious play process, where players do not have access to each other's payoff matrix. Each has to observe the other's actions up to present and plays the action generated based on the best response to these observations. In a regular fictitious play process, each player makes a maximum likelihood estimate of her opponent's mixed strategy, which results in a time-varying update based on the previous estimate and current action. In this paper, we explore an alternative scheme for frequency update, whose mean dynamic is instead time-invariant. We examine convergence properties of the mean dynamic of the fictitious play process with such an update scheme, and establish local stability of the equilibrium point when both players are restricted to two actions. We also propose an adaptive algorithm based on this time-invariant frequency update.

\end{abstract}

\section{Introduction}

Game theory has recently been used as an effective tool to model and solve many security problems in computer and communication networks. In a noncooperative matrix game between an Attacker and a Defender, if the payoff matrices are assumed to be known to both players, each player can compute the set of Nash equilibria of the game and play one of these strategies to maximize her expected gain (or minimize its expected loss). However, in practice, the players do not necessarily have full knowledge of each other's payoff matrix. For repeated games, a mechanism called fictitious play (FP) can be used for each player to learn her opponent's motivations. In a FP process, each player observes all the actions and makes estimates of the mixed strategy of her opponent. At each stage, she updates this estimate and plays the pure strategy that is the best response (or generated based on the best response) to the current estimate of the other's mixed strategy. It can be seen that in a FP process, if one player plays a fixed strategy (either of the pure or mixed type), the other player's sequence of strategies will converge to the best response to this fixed strategy. Furthermore, it has been shown that, for many classes of games, such a FP process will finally render both players playing a Nash equilibrium (NE).

Specifically, we examine a two-player game, where an Attacker (denoted as player $1$ or $P_1$) and a Defender (denoted as player $2$ or $P_2$) participate in a discrete-time repeated nonzero-sum matrix game. In a general setting, the Attacker has $m$ possible actions and the Defender has $n$ posssible actions to choose from. For example, when $m=n=2$, the Attacker's actions could be to attack one node in a two-node network, and those of the Defender are to defend one of these two nodes. Players do not have access to each other's payoff function. They adjust their strategies based on each other's actions which they observe.

In a stochastic FP process, each player makes a maximum likelihood estimation of her opponent's mixed strategy. As will be seen later on, this will result in a time-varying update of the opponent's empirical frequency, where the weight of the action at time step $k$ is $1/k$. In a practical repeated security game, however, we notice a couple of possible complications. First, players may not have the exact and synchronized time steps. Second, each player may want to adjust the weight of the other's current action to converge either faster or more accurately to the equilibrium. A more flexible scheme to update the estimate of the mixed strategy may be needed in such situations. Motivated by these practical considerations, we examine in this paper a time-invariant frequency update mechanism for fictitious play. Also, as a side note, such a time-invariant update mechanism will allow us to use the analysis tools applicable only to time-invariant systems. 

Security games have been examined extensively in a large number of papers, see for example, \cite{alpcancdc03ids,alpcancdc04,alpcanisdg06,chenphd,sallphd,gamnets-kien}. Surveys on applications of game theory to network security can be found in \cite{LCA-BOOK-2006-001}, \cite{royshan}. Relevant literature on fictitious play can be found in  \cite{robinson, Miyasawa, uberger, sa04, sa05, MSA05,icc09-kien,acc10nab}. A comprehensive exposition of learning in games can be found in \cite{fule}.

The rest of this paper is organized as follows. In Section \ref{sec:tifp}, we provide an overview of the static game and the standard stochastic FP process, and then introduce the stochastic FP with time-invariant frequency update. The analysis for FP with time-invariant frequency update is given in Section \ref{sec:analysis}. In Section \ref{sec:atifp}, we introduce an adaptive algorithm based on the time-invariant FP process. Next, simulation results are given in Section \ref{sec:sim}. Finally, some concluding remarks will end the paper.
\section{Fictitious Play with Time-Invariant Frequency Update} \label{sec:tifp}
In this Section, we present first an overview of a two-player static games, then the concept of Stochastic Fictitious Play with Time-Varying Frequency Update (TVFU-FP) \cite{sa04,sa05,icc09-kien,acc10nab}, and finally the concept of Stochastic Fictitious Play with Time-Invariant Frequency Update (TIFU-FP). While we introduce both classical version and stochastic version of static games, we restrict ourseves to only stochastic fictitious play in Subsections \ref{ss:tvfp} and \ref{ss:tifp} and in the rest of the paper. 
\subsection{Static Games} \label{static}
We consider here static security games, where each player $P_i,\ i=1,2$, has two possible actions (or pure strategies). We use $v_i$, to denote the action of $P_i$. Let $\Delta(2)$ be the simplex in $\Re^2$, i.e., 
  \begin{equation}
  \Delta(2) \equiv \left\{ s \in \Re^2 \vert s_1,s_2 \geq 0 \textrm{ and } s_1+s_2 = 1 \right\}.
  \end{equation}
  Each $v_i$ takes value in the set of (two) vertices of $\Delta(2)$: $v_i=[1 \ 0]^T$ for the first action, and $v_i=[0 \ 1]^T$ for the second action. In a static game, player $P_i$ selects an action $v_i$ according to a mixed strategy $p_i \in \Delta(2)$. The (instant) payoff for player $P_i$ is\footnote{As standard in the game theory literature, the index $-i$ is used to indicate those of other players, or the opponent in this case.} $v^T_i M_i v_{-i} + \tau_i H(p_i)$, where $M_i$ is the payoff matrix of $P_i$, and $H(p_i)$ is the entropy of the probability vector $p_i$, $H(p_i)=-p_i^T log(p_i)$. The weighted entropy $\tau_i H(p_i)$ with $\tau_i \geq 0$ is introduced to boost mixed strategies. In a security game, $\tau_i$ signifies how much player $i$ wants to randomize its actions, and thus is not necessarily known to the other player. Also, for $\tau_1=\tau_2=0$ (referred to as classical FP), the best response mapping can be set-valued, while it has a unique value when $\tau_i>0$ (referred to as stochastic FP). For a pair of mixed strategy $(p_1,p_2)$, the utility functions are given by the expected payoffs:
   \begin{eqnarray} \label{utility}
\nonumber  U_i(p_i,p_{-i}) &=& E \left[ v^T_i M_i v_{-i} \right] + \tau_i H(p_i) \\
     &=& p_i^T M_i p_{-i} + \tau_i H(p_i).
  \end{eqnarray}
  Now,	the \textit{best response} mappings $\beta_i: \Delta(2) \rightarrow  \Delta(2)$ are defined as:
  \begin{equation} \label{utility_max}
  \beta_i(p_{-i}) = \arg \max_{p_i \in \Delta(2)} {U_i(p_i,p_{-i})}.
  \end{equation}
  If $\tau_i>0$, the best response is unique as mentioned earlier, and is given by:
  \begin{equation} \label{best_response}
	\beta_i(p_{-i}) = \sigma \left( \frac{M_i p_{-i}}{\tau_i} \right),
  \end{equation}
  where the soft-max function $\sigma: \Re^2 \rightarrow \ \textrm{Interior}(\Delta(2))$ is defined as
  \begin{equation} \label{softmax}
  (\sigma(x))_j = \frac{e^{x_j}} {e^{x_1}+ e^{x_2}}, j=1,2.
  \end{equation}
  Note that $(\sigma(x))_j > 0$, and thus the range of the soft-max function is just the interior of the simplex.
  
  Finally, a (mixed strategy) Nash equilibrium is defined to be a pair $(\bar{p}_1, \bar{p}_2) \in \Delta(2) \times \Delta(2)$ such that for all $p_i \in \Delta(2)$
  \begin{equation} \label{saddle_point}
  U_i(p_i,\bar{p}_{-i}) \leq U_i(\bar{p}_i,\bar{p}_{-i}).
  \end{equation}
  We can also write a Nash equilibrium $(\bar{p}_1, \bar{p}_2)$ as the fixed point of the best response mappings:
  \begin{equation} \label{best_response_mapping}
  \bar{p}_i = \beta _i (\bar{p}_{-i}), \ i=1,2.
  \end{equation}  
\subsection{Stochastic Fictitious Play with Time-Varying Frequency Update}  \label{ss:tvfp}
From the static game described in Subsection \ref{static}, we define the discrete-time TVFU-FP as follows. Suppose that the game is repeated at times $k \in \left\{0,1,2,\ldots \right\}$. The empirical frequency $q_i(k)$ of player $P_i$ is given by
\begin{equation} \label{em_freq}
q_i(k+1) = \frac{1}{k+1} \sum_{j=0}^k v_i(j).
\end{equation}
Using induction, we can prove the following recursive relation:
\begin{equation} \label{rec_em_freq}
q_i(k+1)  =  \frac{k}{k+1} q_i(k) + \frac{1}{k+1} v_i(k).
\end{equation}
From the equations of discrete-time TVFU-FP (\ref{em_freq}), (\ref{rec_em_freq}), the continuous-time version of the iteration can be written down as follows \cite{icc09-kien}:
\begin{eqnarray} \label{ctfp}
 \dot{p}_i(t) &=& \beta_i(p_{-i}(t)) -  p_i(t), \ i=1,2.
\end{eqnarray}

\subsection{Stochastic Fictitious Play with Time-Invariant Frequency Update}  \label{ss:tifp}

In TVFU-FP, players take the maximum likelihood estimate of the mixed strategy of their opponent (\ref{em_freq}), (\ref{rec_em_freq}). In TIFU-FP, the estimates of the mixed strategies will be calculated in a time-invariant manner as follows:
\begin{eqnarray} 
\label{tiest1} r_i(1) &=& v_i(0), \\
\label{tiest2} r_i(k+1)  &=&  (1-\eta) r_i(k) + \eta v_i(k),
\end{eqnarray}
where $\eta$ is a constant and $0 < \eta < 1$. For each player, this is basically the exponential smoothing formula used in time series analysis (See for example \cite{kellerwarrack00}). We will prove that with this formulation, at time $k$, $r_i(k)$ will be a weighted average of all the actions up to present of player $i$ where more recent actions have higher weights.
Suppose that the payoff matrices of player $1$ and player $2$ are, respectively,
\begin{equation}
M_1 = \left(
\begin{array}{cc}
	a	&	b \\
	c	&	d \\
\end{array}
\right), \  \
M_2 = \left(
\begin{array}{cc}
	e	&	g \\
	f	&	h \\
\end{array}
\right).
\end{equation}
\begin{assumption} \label{as:payoff}
Based on a realistic security game, we can make the following assumptions:
\begin{itemize}
\item $a<c$: When the Defender defends, the payoff of the Attacker will be decreased if it attacks.
\item $b>d$: When the Defender does not defend, the payoff of the Attacker will be increased if it attacks.
\item $e>f$: When the Attacker attacks, the payoff of the Defender will be decreased if it does not defend.
\item $g<h$: When the Attacker does not attack, the payoff of the Defender will be increased if it does not defend.
\end{itemize}
\end{assumption}
\begin{algorithm}[htp]
\caption{Fictitious Play with Time-Invariant Frequency Update.} \label{alg:tifp}
\begin{algorithmic}[1]
\STATE Given payoff matrix $M_i$, coefficient $\tau_i>0$, $i=1,2$.
\FOR{$k \in \left\{0,1,2,\ldots \right\}$}
\STATE Update the estimated frequency of the opponent using (\ref{tiest1}), (\ref{tiest2}).
\STATE Compute the best response using (\ref{best_response}). (Note that the result is always a completely mixed strategy.)
\STATE Randomly play an action $v_i(k)$ according to the best response mixed strategy $\beta_i(r_{-i}(k))$.
 \ENDFOR
\end{algorithmic}
\end{algorithm}
In TIFU-FP, both players employ Algorithm \ref{alg:tifp}. The mean dynamic of the evolution of TIFU-FP can be written as:
\begin{eqnarray} \label{est_mean_dynamic}
r_i(k+1)  =  (1-\eta) r_i(k) + \eta \beta_i(r_{-i}(k)), \ i=1,2.
 \end{eqnarray} 
Note that Equations (\ref{est_mean_dynamic}) are just evolution of the \textit{estimated frequencies}; the \textit{empirical frequencies} still evolve in a time-varying manner:
\begin{eqnarray} 
q_i(k+1)  =  \frac{k}{k+1} q_i(k) + \frac{1}{k+1} v_i(k), \ i=1,2.
\end{eqnarray}
The mean dynamic of empirical frequencies then can be written as 
\begin{eqnarray}
q_i(k+1)  =  \frac{k}{k+1} q_i(k) + \frac{1}{k+1} \beta_i(r_{-i}(k)),\  i=1,2.
\end{eqnarray}
\section{Analysis} \label{sec:analysis}
\subsection{Nash Equilibrium of the Static Game}
We start the analysis with the following result for the static games given in Subsection \ref{static}.
\begin{proposition}
The static $2$-player $2$-action game in \ref{static} with Assumption \ref{as:payoff} and $\tau_1,\tau_2>0$ admits a unique Nash equilibrium.
\end{proposition}
\begin{figure}[ht]
\centering
  \includegraphics[width=8cm]{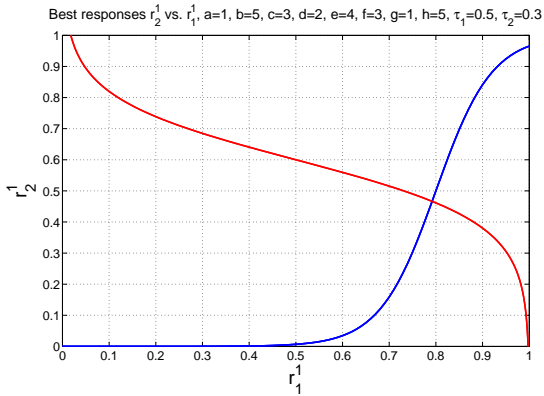} \\
  \caption{Static $2$-player $2$-action game in \ref{static} with Assumption \ref{as:payoff} and $\tau_1,\tau_2>0$ - Best response mappings.}
   \label{fig:sfp_best_response} 
\end{figure}
\begin{proof}
In what follows, let $r_1 \equiv (r^1_1, r^2_1)^T$, $r_2 \equiv (r^1_2, r^2_2)^T$, $\beta_1(r_2) \equiv (\beta_1^1(r_2), \beta_1^2(r_2))^T$, and $\beta_2(r_2) \equiv (\beta_2^1(r_2), \beta_2^2(r_2))^T$.
We first use the Brouwer fixed point theorem (see for example \cite{basols99}) to prove the existence of a Nash equilibrium, then use monotonicity of $\beta^1_1(r_2^1)$ and $\beta^1_2(r_1^1)$ to prove that the fixed point is unique. Here we write $\beta^1_1(r_2^1)$ as a scalar-valued function of the only independent variable $r_2^1$ using the fact that $r_2 \in \Delta(2)$, or $r_2^1+r_2^2=1$. Function $\beta^1_2(r_1^1)$ is defined similarly. Also, as $r_1,r_2 \in \Delta(2)$, a pair $(r^1_1, r^1_2)$ completely specifies the estimated frequencies of the players.
As seen from Equation (\ref{best_response_mapping}), a Nash equilibrium $(\bar{r}_1, \bar{r}_2)$ is a fixed point of the best response mapping:
  \begin{eqnarray} 
\nonumber   r_i &=& \beta_i (r_{-i}), \ i=1,2.
  \end{eqnarray} 
It suffices to write this mapping as $r = \beta (r)$, where $r = (r^1_1, r^1_2)^T$. Specifically, the mapping $\beta$ can be detailed as:
\begin{eqnarray}
\nonumber r^1_1 &=& \beta^1_1(r_2^1)=\left( \sigma \left( \frac{M_1r_2}{\tau_1} \right) \right)_1  \\
\nonumber &=& \frac{e^{ \left\{ \frac{1}{\tau_1} [a r^1_2 + b(1-r^1_2)] \right\} }} 
{e^ {\left\{ \frac{1}{\tau_1} [a r^1_2 + b(1-r^1_2)] \right\}}+e^{ \left\{ \frac{1}{\tau_1} [c r^1_2 + d(1-r^1_2)] \right\}}}.
\end{eqnarray}
It can be seen that $\beta^1_1(r_2^1) \in (0,1)$. Similarly, we have
\begin{eqnarray}
\nonumber r^1_2 &=& \beta^1_2(r_1^1) \in (0,1). 
\end{eqnarray}
Thus $\beta$ is a transformation from $[0,1]^2$ to $[0,1]^2$, which is a compact convex set. As both mappings $\beta^1_1$ and $\beta^1_2$ that constitute $\beta$ are continuous, $\beta$ is also a continuous transformation. Using the Brouwer fixed point theorem, there exists at least one fixed point $\bar{r}$ such that $\bar{r} = \beta (\bar{r})$, which is a Nash equilibrium of the static game. Now we examine the derivatives of $\beta^1_1(r_2^1)$ and $\beta^1_2(r_1^1)$ with respect to their own independent variables:
\begin{eqnarray}
\nonumber \frac{ d\beta^1_1(r_2)}{d r^1_2} &=& \frac{1}{\tau_1} [(a-c)+(d-b)] \beta^1_1(r_2) \beta^2_1(r_2),\\
\nonumber \frac{ d\beta^1_2(r_1)}{d r^1_1} &=& \frac{1}{\tau_2} [(e-f)+(h-g)]\beta^2_1(r_1) \beta^2_2(r_1).
\end{eqnarray}
From Assumption \ref{as:payoff}, $(a-c)+(d-b)<0$ and $(e-f)+(h-g)>0$. Thus $\beta^1_1(r_2^1)$ is strictly decreasing in $r^1_2$, and $\beta^1_2(r_1^1)$ is strictly increasing in $r^1_1$. Now suppose that there exist two distinct Nash equilibria, $(r'_1, r'_2)$ and $(r''_1, r''_2)$. Obviously, $r'_1 \neq r''_1$, otherwise we will have $r'_2 = r''_2$, and these two points coincide. Without loss of generality, assume that $r'_1 < r''_1$. As $\beta^1_2(r_1^1)$ is strictly increasing in $r^1_1$, we have that $r'_2 < r''_2$. However, $\beta^1_1(r_2^1)$ is strictly decreasing in $r^1_2$, so $r'_1 > r''_1$, which is contradictory to the initial assumption. Thus the Nash equilibrium is unique.
\end{proof}
We illustrate in Figure \ref{fig:sfp_best_response} the curves $\beta^1_1(r_2^1)$ and $\beta^1_2(r_1^1)$ with the values of $M_1$, $M_2$, $\tau_1$, and $\tau_2$ as shown. The intersection of these two curves is the Nash equilibrium of the static game.
\subsection{Estimated Frequencies and Empirical Frequencies}
We present here two propositions for TIFU-FP: The first shows the weights of each player's actions in the estimated frequency, and the second shows the relationship between estimated frequencies and empirical frequencies.
\begin{proposition}
For $k \geq 2$, the estimated frequencies in TIFU-FP constructed using (\ref{tiest1}), (\ref{tiest2}) will satisfy
\begin{eqnarray}
\nonumber r_i(k) &=& (1-\eta)^{k-1} v_i(0) + (1-\eta)^{k-2} \eta v_i(1) \\
\nonumber &&+  (1-\eta)^{k-3} \eta v_i(2) + \ldots + (1-\eta) \eta v_i(k-2) \\
\label{rik} &&+ \eta v_i(k-1),
\end{eqnarray}
where $i=1,2$.
\end{proposition}
\begin{proof}
This result can be proved using induction.
\end{proof}
\begin{proposition}
In TIFU-FP, the empirical frequencies are related to the estimated frequencies calculated using (\ref{tiest1}), (\ref{tiest2}) through the following equation:
\begin{eqnarray} \label{empest}
\nonumber q_i(k+1) = \frac{1}{k+1} && \left(  \frac{2 \eta -1}{\eta} r_i(1) + r_i(2) + \ldots + r_i(k) \right. \\
&& \left. + \frac{r_i(k+1)}{\eta} \right), \ i=1,2.
 \end{eqnarray}
\end{proposition} 
\begin{proof}
This result can be proved by writing the actions of player $P_i$ at times $0,1,\ldots,k$ in terms of the estimated frequencies at times $1,2,\ldots,(k+1)$.
\end{proof}
\subsection{Convergence Properties of the Mean Dynamic in TIFU-FP} 
\begin{theorem} \label{thm:conv}
Consider a TIFU-FP with Assumption \ref{as:payoff} and $\tau_1,\tau_2>0$. The mean dynamic  given in Equations (\ref{est_mean_dynamic}) is asymptotically stable if and only if
\begin{equation} \label{conv_tifufp}
\eta < \frac{2}{\frac{\left[(c-a)+(b-d)\right]\left[(e-f)+(h-g)\right]}{\tau_1 \tau_2} \bar{r}^1_1 \bar{r}^2_1 \bar{r}^2_1 \bar{r}^2_2 +1 }.
\end{equation}
\end{theorem} 
\vspace{5 mm}
\begin{proof}
As can be seen in Equations (\ref{est_mean_dynamic}), this is a deterministic nonlinear discrete-time time-invariant system. We linearize the system at the fixed point and examine stability properties of the linearized system using techniques described in standard textbooks for nonlinear systems (e.g., \cite{KHA}). Using the mean dynamic (\ref{est_mean_dynamic}), where
\begin{equation}
r_1(k) = \left( 		
	\begin{array}{c}
			r^1_1(k)\\
			r^2_1(k)\\
		\end{array}
	\right),
r_2(k) = \left( 		
	\begin{array}{c}
			r^1_2(k)\\
			r^2_2(k)\\
		\end{array}
	\right),
\end{equation}
it can be seen that a pair $(\bar{r}_1,\bar{r}_2)$ that satisfies $\bar{r}_i = \beta_i(\bar{r}_{-i}), \ i=1,2$, is a fixed point of the system.
Consider the Jacobian matrix
\begin{equation}
\nonumber J=\frac{\partial F(r)}{\partial r} = 
 \left( 		
	\begin{array}{cc}
			\frac{\partial F_1(r)}{\partial r^1_1} & \frac{\partial F_1(r)}{\partial r^1_2}\\
			\frac{\partial F_2(r)}{\partial r^1_1} & \frac{\partial F_2(r)}{\partial r^1_2}\\
		\end{array}
	\right).
\end{equation}	
We have that
\begin{eqnarray}
\nonumber \frac{\partial F_1(r)}{\partial r^1_1} &=& \frac{\partial F_2(r)}{\partial r^1_2} = 1-\eta, \\
\nonumber \frac{\partial F_1(r)}{\partial r^1_2} &=& \eta \frac{ d\beta^1_1(r_2)}{d r^1_2}.
\end{eqnarray}
Recall that $\beta_1(r_2)=\sigma \left( \frac{M_1r_2}{\tau_1} \right)$, where
\begin{equation}
\nonumber \frac{M_1r_2}{\tau_1}= \left( 	\begin{array}{c}
			\frac{1}{\tau_1} [a r^1_2 + b(1-r^1_2)] \\
			\frac{1}{\tau_1} [c r^1_2 + d(1-r^1_2)] \\
		\end{array}
    \right).
\end{equation}
Thus
\begin{equation}
\nonumber \beta^1_1(r_2)=\frac{e^{ \left\{ \frac{1}{\tau_1} [a r^1_2 + b(1-r^1_2)] \right\} }} 
{e^ {\left\{ \frac{1}{\tau_1} [a r^1_2 + b(1-r^1_2)] \right\}}+e^{ \left\{ \frac{1}{\tau_1} [c r^1_2 + d(1-r^1_2)] \right\}}}.
\end{equation}
Then
\begin{equation}
\nonumber \frac{ d\beta^1_1(r_2)}{d r^1_2}= \frac{1}{\tau_1} [(a-c)+(d-b)] \beta^1_1(r_2) \beta^2_1(r_2),
\end{equation}
\begin{equation}
\nonumber \frac{\partial F_1(r)}{\partial r^1_2}= \frac{\eta}{\tau_1} [(a-c)+(d-b)] \beta^1_1(r_2) \beta^2_1(r_2).
\end{equation}
At the fixed point $(\bar{r}_1,\bar{r}_2)$, we can write
\begin{equation}
\nonumber \frac{\partial F_1(\bar{r})}{\partial r^1_2}= \frac{\eta}{\tau_1} [(a-c)+(d-b)] \bar{r}^1_1 \bar{r}^2_1.
\end{equation}
Similarly,
\begin{equation}
\nonumber \frac{\partial F_2(\bar{r})}{\partial r^1_1}= \frac{\eta}{\tau_2} [(e-f)+(h-g)]\bar{r}^1_2 \bar{r}^2_2 .
\end{equation}
Using the conditions for local stability, $\left| \mu_{1,2} \right| \leq 1$, where $\mu_{1,2}$ are eigenvalues of the Jacobian matrix, we finally have the condition in Equation (\ref{conv_tifufp}).
\end{proof}
\begin{remark}
Although this theorem only mentions the asymptotic stability of the estimated frequencies (of the mean dynamic), once these estimated frequencies converge to the Nash equilibrium, the best responses will also converge to the Nash equilibrium, and so will the empirical frequencies in the long run.
\end{remark}
\section{Adaptive Fictitious Play} \label{sec:atifp}
In this section we examine an adaptive FP algorithm (hereafter referred to as AFP) based on FP with Time-Invariant Frequency Update, where the step size $\eta$ is piecewise constant and decreased over time. 
For the specific implementation shown in Algorithm~\ref{alg:atifp}, the step size is either kept fixed or halved, based on the variance of empirical frequency in the previous time window.
\begin{algorithm}[htp]
\caption{Adaptive Fictitious Play} \label{alg:atifp}
\begin{algorithmic}[1]
\STATE Given payoff matrix $M_i$, coefficient $\tau_i$, $i=1,2$, initial step size $\eta_0$, minimum step size $\eta_{min}$, and window size $T$.
\FOR{$k \in \left\{0,1,2,\ldots \right\}$}
\STATE Update the estimated frequency of the opponent, $r_{-i}$, using (\ref{tiest1}), (\ref{tiest2}).
\STATE Compute the best response  mixed strategy $\beta_i(r_{-i}(k))$ using (\ref{best_response}).
\STATE Randomly play an action $a_i(k)$ according to the best response mixed strategy $\beta_i(r_{-i}(k))$,
such that the expectation $E \left[a_i(k) \right]=\beta_i(r_{-i}(k))$.
\IF{at the end of a time window, $\mod(k,T)=0$, } 
\STATE Compute the standard deviation of the estimated frequencies ($\rm{stdef}$)
in the time window $[r_{-i}(k-T+1),\ldots,r_{-i}(k)]$ (using an unbiased estimator): 
\begin{eqnarray}
\nonumber \rm{mef}(k) &=& \frac{1}{T}\sum_{h=k-T+1}^k r_{-i}(h)  \\
\rm{stdef}(k) &=&
\nonumber \sqrt{\frac{\sum_{h=k-T+1}^k \left( r_{-i}(h) - \rm{mef}(k) 
   \right)^2} {(T-1)}}
\end{eqnarray}   
\IF{the computed $\rm{stdef}(k)$ has decreased compared to previous time window} 
\STATE Decrease step size: $\eta=0.5\, \eta$ and $\eta=\max (\eta, \eta_{min})$.
\ELSE
\STATE Keep step size $\eta$ constant.
\ENDIF
\ENDIF
\ENDFOR
\end{algorithmic}
\end{algorithm}
\section{Simulation results} \label{sec:sim}
We present in this section some simulation results for TIFU-FP and AFP where the payoff matrices and entropy coefficients are chosen to be
 \begin{equation} \nonumber
M_1 = \left(
\begin{array}{cc}
	1	&	5 \\
	3	&	2 \\
\end{array}
\right), \  
M_2 = \left(
\begin{array}{cc}
	4	&	1 \\
	3	&	5 \\
\end{array}
\right), \ 
\tau_1=0.5, \ \tau_2=0.3. 
\end{equation}
The Nash Equilibrium of the static game is $(0.79, \ 0.21)$ and $(0.47, \ 0.53)$. The local stability threshold (the RHS of Equation (\ref{conv_tifufp})) is $\eta_0 = 0.2536$. For simplicity, in the graphs shown here, we only plot the first component of each frequency vector.
\subsection{Fictitious Play with Time Invariant Frequency Update} \label{ss:simtifp}
Some simulation results for the mean dynamic of TIFU-FP (Equations (\ref{est_mean_dynamic})) are given in Figures \ref{fig:tisfp_est_freq_025} and \ref{fig:tisfp_est_freq_026}. When $\eta=0.25<\eta_0 = 0.2536$, the estimated frequencies are shown in Figure \ref{fig:tisfp_est_freq_025}. The simulation results show that both estimated frequencies and empirical frequencies (not presented here due to space limitations) converge to the NE as expected. When $\eta=0.26>\eta_0$, however, the estimated frequencies do not converge anymore. These simulations thus confirm the theoretical result in Theorem \ref{thm:conv}. It is also worth noting that the empirical frequencies in the case $\eta=0.26$ still converge to the NE.
\begin{figure}
\centering
  \includegraphics[width=8cm]{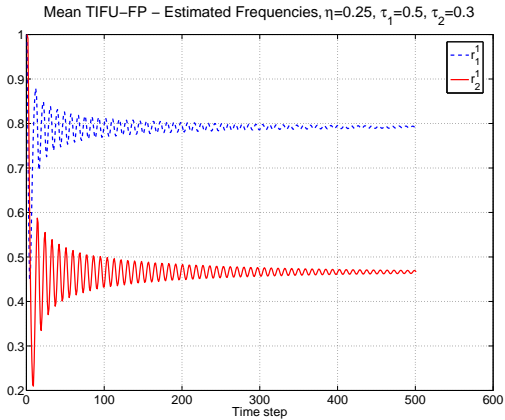} \\
  \caption{Mean dynamic of FP with Time-Invariant Frequency Update - Estimated Frequencies, $\eta=0.25$, $\eta_0 = 0.2536$.} \label{fig:tisfp_est_freq_025} 
\end{figure}
\begin{figure}
\centering
  \includegraphics[width=8cm]{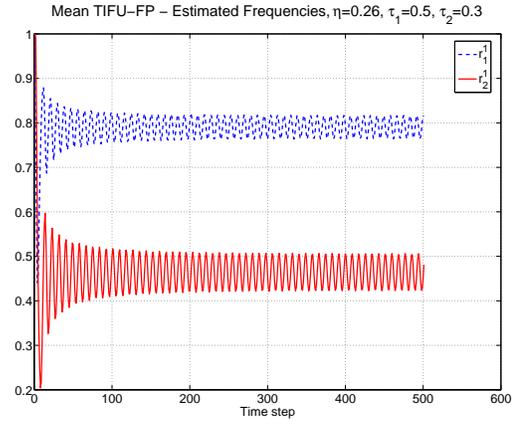} \\
  \caption{Mean dynamic of FP with Time-Invariant Frequency Update - Estimated Frequencies, $\eta=0.26$, $\eta_0 = 0.2536$.} \label{fig:tisfp_est_freq_026} 
\end{figure}
Unlike the mean dynamic, a stochastic TIFU-FP process (generated with Algorithm \ref{alg:tifp}) exhibits significant random fluctuations. The graph in Figure \ref{fig:tisfp_s_est_freq_001} shows the estimated frequencies of such a process where we choose $\eta=0.01$. However, the empirical frequencies (whose graph is not shown here due to space limitations) still converge to the NE .
\begin{figure}
\centering
  \includegraphics[width=8cm]{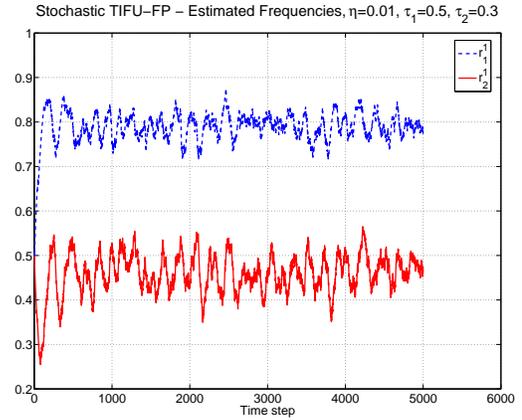} \\
  \caption{Stochastic FP with Time-Invariant Frequency Update - Estimated Frequencies, $\eta=0.01$.} \label{fig:tisfp_s_est_freq_001} 
\end{figure}
\begin{figure}
\centering
  \includegraphics[width=8cm]{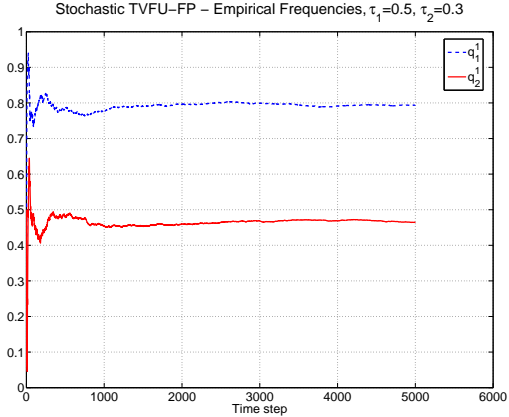} \\
  \caption{Stochastic FP with Time-Varying Frequency Update - Empirical Frequencies.} \label{fig:tvsfp_emp_freq} 
\end{figure}
\subsection{Adaptive Fictitious Play}
Some simulation results for adaptive FP are shown in Figures \ref{fig:atisfp_emp_freq} and \ref{fig:atisfp_step_size}. The payoff matrices and entropy coefficients are the same as those in \ref{ss:simtifp}. Initial and minimum step sizes are chosen to be $\eta_0=0.1$ and $\eta_{min}=0.0005$, respectively. The time window for updating the step size is $T=50$ steps. The evolution of the empirical frequencies are depicted in Figure~\ref{fig:atisfp_emp_freq}, which shows that adaptive
FP converges faster than the stochastic FP with time-varying frequency update (TVFU-FP) (Figure \ref{fig:tvsfp_emp_freq}). We however remark that it is possible to incorporate a decreasing coefficient into the step size in  TVFU-FP (which is originally $1/k$) to make the TVFU-FP process converge faster \cite{MSA05}. The update of the step size in adaptive FP is shown in Figure~\ref{fig:atisfp_step_size}. Note that when compared to the step size $1/k$ in TVFU-FP, the step sizes in adaptive FP are higher in the beginning and smaller afterwards, resulting in aggressive convergence first and less fluctuation in the stable phase. 
\begin{figure}
\centering
  \includegraphics[width=8cm]{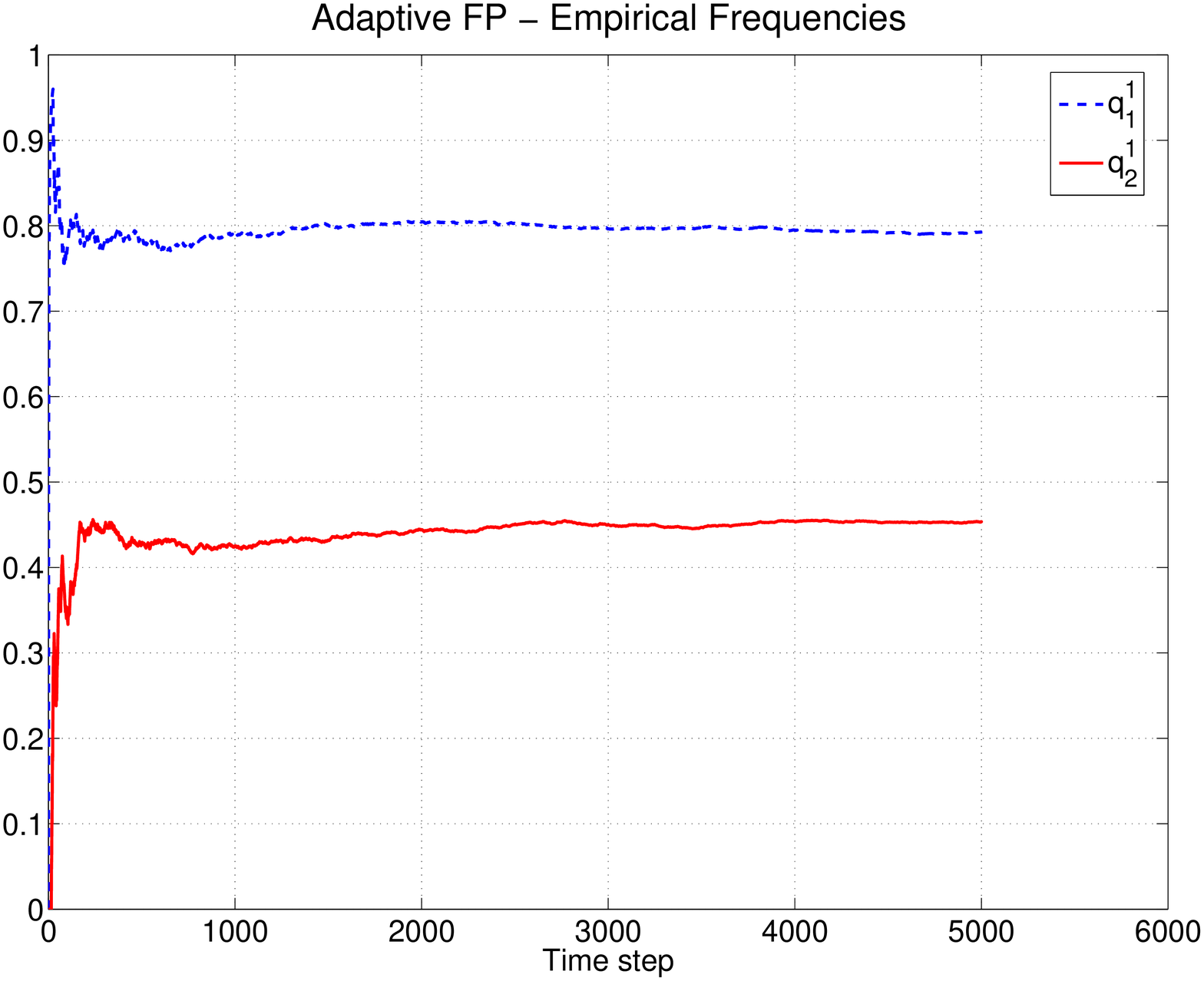} \\
  \caption{Adaptive Stochastic FP - Empirical Frequencies.} \label{fig:atisfp_emp_freq} 
\end{figure}
\begin{figure}
\centering
  \includegraphics[width=8cm]{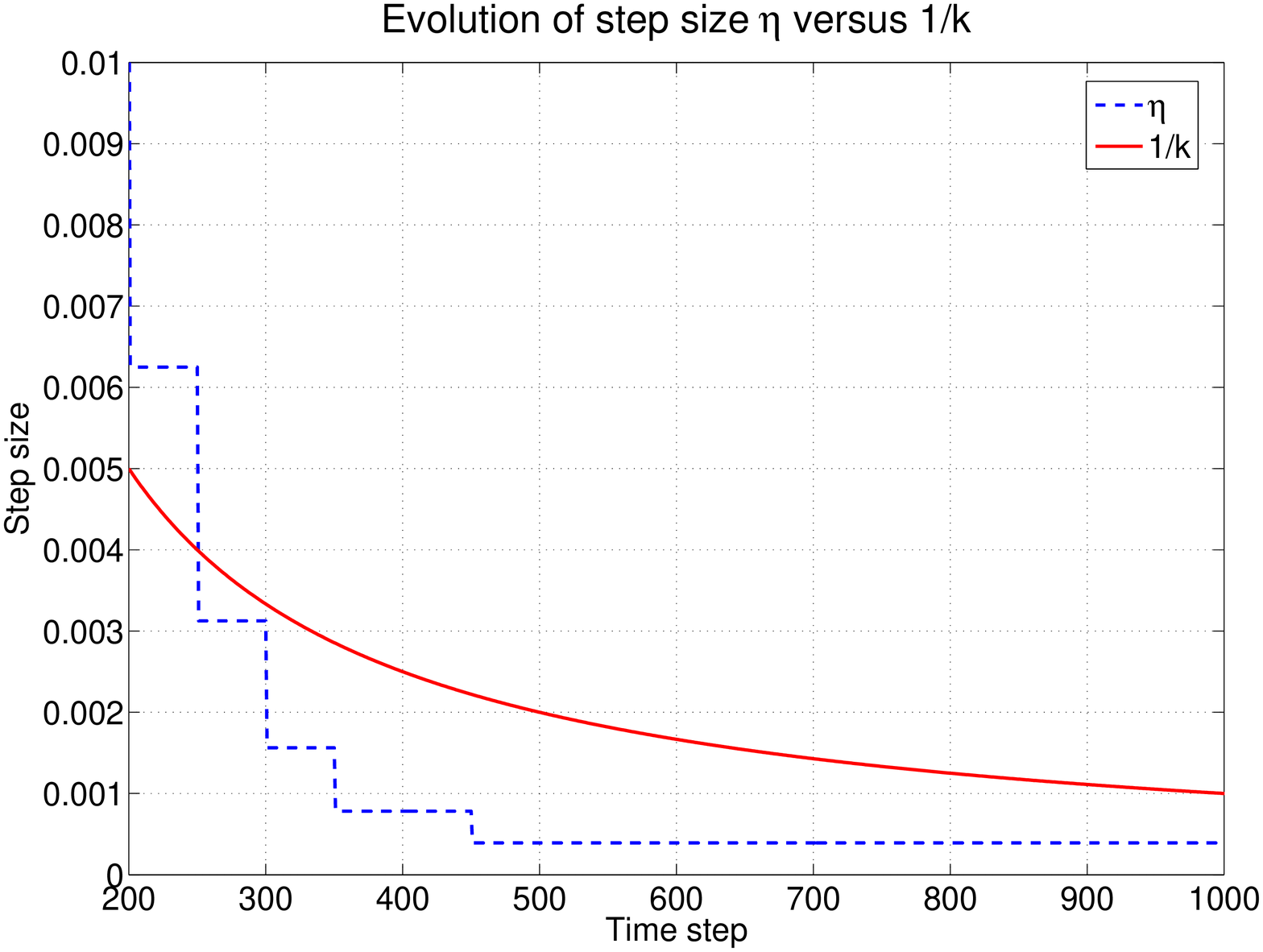} \\
  \caption{Adaptive Stochastic FP - Evolution of step size.} \label{fig:atisfp_step_size} 
\end{figure}
\section{Conclusions}
In this paper, we have introduced a time-invariant scheme to estimate the frequency of the opponent's actions in a two-player two-action fictitious play process. We have proved local stability of the unique Nash equilibrium for the mean version of this FP dynamic. This frequency update scheme, when used adaptively, allows players to converge faster to the Nash equilibrium. For this two-player two-action FP, conditions for global stability, if they exist, are yet to be found. Also, having more than two possible actions for each player is an intriguing research extension.
%
\bibliographystyle{abbrv}

\end{document}